\documentclass{elsarticle}
\usepackage[utf8]{inputenc}
\usepackage[T1]{fontenc}
\usepackage[english]{babel}
\usepackage[standard]{ntheorem}
\usepackage{amsmath,amssymb,amstext}
\usepackage{mathabx}
\usepackage{color}

\begin{document}
\begin{frontmatter}

\title{Chaos in DNA Evolution}

\author[femto,ufc]{Jacques M. Bahi}
\ead{jacques.bahi@univ-fcomte.fr}

\author[femto,ufc]{Christophe Guyeux\corref{cor1}}
\ead{christophe.guyeux@univ-fcomte.fr}

\author[chrono,ufc]{Antoine Perasso\corref{cor1}}
\ead{christophe.guyeux@univ-fcomte.fr}

\address[chrono]{Chrono-environnement laboratory, UMR 6249 CNRS} 
\address[femto]{FEMTO-ST Institute, UMR 6174 CNRS} 
\address[ufc]{University of Franche-Comt\'{e}, Besan\c con, France}

%\date{Received: date / Accepted: date}
% The correct dates will be entered by the editor

\begin{abstract}
In this paper, we explain why the chaotic model (CM) of Bahi and Michel (2008) accurately simulates gene mutations over time. First, we demonstrate that the CM model is a truly chaotic one, as defined by Devaney. Then, we show that mutations occurring in gene mutations have the same chaotic dynamic, thus making the use of chaotic models relevant for genome evolution.
\end{abstract}
\end{frontmatter}

\section{Introduction}
Codons are not uniformly distributed in the genome.
Over time, mutations have introduced some variations in their
frequency of apparition.
It can be attractive to study the genetic patterns (blocs of more than one nucleotide: dinucleotides, trinucleotides...) that appear and disappear depending on mutation parameters.
Mathematical models allow the prediction of such an evolution, in such a way
that statistical values observed in current genomes can be recovered
from hypotheses on past DNA sequences.
A first model for genome evolution was proposed in 1969 by Thomas
Jukes and Charles Cantor \cite{Jukes69}. This first model is very simple,
as it supposes that each nucleotide $A,C,G,T$ has the probability $m$ to
mutate to any other nucleotide, as described in the following mutation
matrix,
$$
\left(\begin{array}{cccc}
1-3m & m & m & m\\
m & 1-3m & m & m\\
m & m & 1-3m & m\\
m & m & m & 1-3m\\
\end{array}\right)
$$
In this matrix, the coefficient in row 3, column 2 represents the
probability that the nucleotide $G$ mutates in $C$ during the next time
interval, \emph{i.e.}, $P(G \rightarrow C)$.
This first attempt has been followed up by Motoo Kimura \cite{Kimura80},
who has reasonably considered that transitions ($A \longleftrightarrow G$ and
$T \longleftrightarrow C$) should not have the same mutation rate as
transversions ($A \longleftrightarrow T$, $A \longleftrightarrow C$, $T
\longleftrightarrow G$, and $C \longleftrightarrow G$), leading to the
following mutation matrix.
$$
\left(\begin{array}{cccc}
1-a-2b & b & a & b\\
b & 1-a-2b & b & a\\
a & b & 1-a-2b & b\\
b & a & b & 1-a-2b\\
\end{array}\right)
$$
This model was refined by Kimura in 1981 (three
constant parameters, to make a distinction between natural
$A\longleftrightarrow T$, $C \longleftrightarrow G$ and unnatural
transversions), Joseph Felsenstein \cite{Felsenstein1980}, Masami Hasegawa,
Hirohisa Kishino, and Taka-Aki Yano \cite{Hasegawa1985} respectively.
The differences between these models are in the number of parameters they use, but
\emph{all of the latter manipulate constant parameters}.
However, they all are rudimentary as they only allow to study nucleotide evolution, not genetic patterns mutations.
From 1990 to 1994, Didier Arquès and Christian Michel proposed models
based on the RY purine/pyrimidine alphabet \cite{Arques1990,Arques1990a,Arques1992,Arques1993a,Arques1993b,Arques1994}.
These models have been abandoned by their own authors in favor of models
over the $\{A,C,G,T\}$ alphabet.
More precisely, in 1998 Didier Arquès, Jean-Paul Fallot, and Christian Michel
proposed a first evolutionary model on the  $\{A,C,G,T\}$
alphabet that is based on trinucleotides~\cite{Arques1998}. With such a
model, the mutation matrix now has a size of $64 \times 64$ (there are 64
trinucleotides). This model comprises 3 parameters $p,q,r$ that correspond,
for a given trinucleotide $XYZ$, to the probability $p$ of mutation of the
first nucleotide $X$, the mutation probability $q$ of $Y$, and the
probability $r$ that $Z$ mutates.
As for the nucleotides based models, this new approach
has only taken into account constant parameters.
In 2004, Jacques M. Bahi and Christian Michel published novel
research work in which the 1998 model was improved by replacing constants parameters by new parameters dependent on time \cite{Bahi2004}.
In this way, it has been possible to simulate a gene evolution that is non-linear.
However, the following years, these researchers returned to models embedding constant parameters, probably due to the fact that the 2004 model leads to poor results: only one of the twelve studied cases allows to recover values that are close to reality.
For instance, in 2006, Gabriel Frey and Christian Michel proposed a model that uses 6 constant parameters \cite{Frey2006a}, whereas in 2007, Christian Michel used a model with 9 constant parameters that generalize those of 1998 and 2006~\cite{Michel2007c}.
Finally, Jacques M. Bahi and Christian Michel have recently introduced
in~\cite{Bahi2008}, a last model with 3 constant parameters,
but \emph{whose evolution matrix evolves over time}.
In other words, trinucleotides that have to mutate (modifying trinucleotide content without changing their location) are not fixed, but they are randomly picked among a subset of potentially mutable trinucleotides.
This model, called ``chaotic model'' (CM), allows good recovery of various
statistical properties detected in the genome.
Furthermore, this model fits well with the hypothesis of some primitive
genes that have mutated over time.

In this paper, we ask why the CM model yields such good results. Obviously,
it is reasonable to assume that not all of the trinucleotides have to mutate each
time as, for instance, the stop codons that have very small
mutation probabilities.
However, such a biological claim is not sufficient to explain the success
of the CM model to accurately simulate the dynamics of
mutations in genomes.
Our proposal is that \emph{the dynamics of genomes evolution is indeed chaotic},
as defined by the Devaney's formulation~\cite{devaney,Banks92}. This is why linear
non-chaotic models of  evolution are far from what they attempt
to model, leading to a poor accuracy in their prediction.
By contrast, we have recently established that discrete dynamical systems
in chaotic iterations satisfy Devaney's definition of chaos~\cite{bahi2010hash}.
Thus the CM model, which is the first mutation model based on chaotic iterations~\cite{Bahi2008} (considering that the set of trinucleotides that can possibly mutate evolve over time), uses a chaotic
dynamical system to describe a chaotic
behavior, leading to a model of the same nature as the phenomenon under study.
%This research work is the first one of a series we intend to write.
We finally demonstrate that, in contrast to inversions, mutations occurring in
genomes have a chaotic dynamics. So at least one type of genomes
reorganization process is chaotic, according to the formulation of such a behavior in the mathematical theory of chaos.
%Then we will show in future works that almost all the other types of modifications
%that occur during genomes evolution are chaotic, which should lead at
%least to the abandonment of linear models for genomes evolution prediction.
%In this paper, we wonder why a model having a chaotic dynamics gives, in a certain way, better results than the standard model, to predict the evolution of genomes through mutations.
%We will show that an important mutation mechanism, namely the inversion, has a chaotic dynamic over time.
%Consequences of this proof, for biology and evolution models, will finally be discussed.

The remainder of this research work is organized as follows.
In Section~\ref{sec:CMModel}, 
the CM model of genomes evolution is 
recalled and its performances are synthesized.
Then, in the next section, basic recalls concerning chaotic iterations and
Devaney's chaos are given.
Genomics mutations are formalized through a discrete dynamical system 
and studied in Section \ref{sec:genmut}. In particular, they are proven to be 
chaotic according to Devaney. 
%Section~\ref{sec:other} investigates o
Other categories of genomics
rearrangements are investigated too, namely transpositions and inversions.
This research work ends with a conclusion section, where the contribution
is summarized and intended future work is listed.

\section{The CM Model of Genome Evolution}
\label{sec:CMModel}
In this section, the CM model is presented, its capability to reasonably approximate mutations into genomes is recalled, and its relationship with chaotic iterations is stated.

\subsection{Gene mutations shown as chaotic iterations}
When considering the model of 2007 with 9 constant parameters~\cite{Michel2007c} that generalizes the models of 1998 and 2006 (\cite{Arques1998} and \cite{Frey2006a} respectively), \emph{all of the trinucleotides have to mutate at each time}.
These models do not take into account the low mutability of the stop codons.
Furthermore, they do not allow mutation strategies to be applied to certain given codons, while the other codons do not mutate.
This is why the model with 3 constant parameters and a chaotic strategy has been proposed in \cite{Bahi2008,Bahi2008a}.
In this model, the set of trinucleotides is divided into two subsets at each time $t$: the first one comprises trinucleotides that can possibly mutate at time $t$, whereas in the second set, trinucleotides cannot change at the considered time.
The trinucleotides that mutate with replacement at time $t$ are randomly picked following a uniform distribution on the set of all possible subsets of trinucleotides (other distribution of probabilities like discrete Poisson process have not been regarded by these authors).
Consequently, the size and the constitution of the subset of mutable trinucleotides change at each time $t$. This subset is denoted by $J(t)$, and this new model has been called ``chaotic model'' CM by the authors of \cite{Bahi2008,Bahi2008a}, as opposed to the former ``standard model'' of 1998~\cite{Arques1998}, due to its relationship with the
chaotic iterations recalled below.
Since the trinucleotides that do not mutate in the chaotic model are not derived from the mutation of other trinucleotides (as, at each iteration, we focus only on the subset of trinucleotides that are allowed to mutate at the considered time), their probabilities of occurrence are constant. Conversely, mutation parameters of the mutable trinucleotides are the same as the 1998 model: $p$, $q$, and $r$ with $p+q+r=1$, for each of the three nucleotide sites.
Let $P_i(t)$ the probability of occurrence of the trinucleotide $i$ at time $t$. Let $A^{(t)}$ be the mutation matrix at time $t$, whose element $(i,j)$ is $P^{(t)}(i \rightarrow j)$: the probability that the $i-$th trinucleotide (ordered in lexicla order) mutates into the $j-$th one. For instance, in line 1 and column 2, there is $P^{(t)}(1 \rightarrow 2) = P^{(t)}(AAA \rightarrow AAC)$. The previous remarks lead to the following formulation:
$$
\left\{
\begin{array}{ll}
P_i'(t) = 0  & \text{if } i \notin J(t),\\
P_i'(t) = \displaystyle{\sum_{j=1}^{64}} (A^{(t)}-I)_{ji}P_j(t) & \text{if } i \in J(t).
\end{array}
\right.
$$
Obviously, this new model is a generalization of the 1998 version. Indeed, if we suppose that $A^{(t)} = A$ for every $t$, then denoting $J(t)$ is the set of all the trinucleotides at time $t$, the latter system can be summarized to its second line, which is exactly the 1998 model.
As the number of mutable trinucleotides changes over time, the mutation matrix is not constant, which leads to the fact that the resolution method used in the standard model cannot be applied here.
To solve the system, authors of \cite{Bahi2008,Bahi2008a} have considered discrete times small enough to be sure that the mutation matrix does not change between $t_i$ and $t_{i+1}$, where the length of $[t_i, t_{i+1}]$ is small enough compared to the mutation rate.
Let $A^{(k)}$ be the (constant) mutation matrix during the time interval $[t_{k-1}, t_k]$.
To compute $P_i'(t_{k-1})$, authors of \cite{Bahi2008,Bahi2008a} have considered that:
$$\dfrac{d~P_i(t_{k-1})}{dt} = \dfrac{P_i(t_k) - P_i(t_{k-1})}{h},$$
\noindent where $h=t_k-t_{k-1}$ is supposed small and constant.
By putting this formula into the previous system, these authors have finally obtained:
$$
\left\{
\begin{array}{ll}
P_i(t_k) = P_i(t_{k-1})  & \text{if } i \notin J(t_k),\\
P_i(t_k) = h \displaystyle{\sum_{j=1}^{64}} (A^{(k)}-I)_{ji}P_j(t_{k-1}) + P_i(t_{k-1}) & \text{if } i \in J(t_k).
\end{array}
\right.
$$
This model has been called the ``discrete time chaotic evolution model CM'' in \cite{Bahi2008,Bahi2008a}.
We will show that this discrete version is, indeed, a gene evolution model that uses chaotic iterations.
To understand the interest of this discrete time chaotic evolution model, we must firstly recall the discovery by Michel \emph{et al.} of a $C^3-$code and its properties~\cite{ArquesMichel1996}.

\subsection{Relevance of the CM model}
A computation of the frequency of each trinucleotide in the 3 frames of genes, in a large gene population (protein coding region) of both eukariotes and prokaryotes, it was established in 1996 that the distribution of trinucleotides in these frames is not uniform~\cite{ArquesMichel1996}.
Such a surprising result has led to the definition of 3 subsets of trinucleotides, denoted by $X_0$, $X_1$, and $X_2$.
These sets are defined as follows.
For each of the 60 trinucleotides different from $AAA$, $TTT$, $CCC$, $GGG$, computes its frequency in
the reading frame $R_0$, in the frame $R_1$ obtained by a shift of 1 nucleotide to the left of $R_0$,
and in the frame $R_2$ obtained by a shift of 2 nucleotides.
If the considered trinucleotide is more frequent in $R_0$ (resp. $R_1, R_2$),
put it in $X_0$ (resp. $X_1, X_2$).
This procedure is repeated, with small variations, until $X_0$, $X_1$, and $X_2$ are respectively made-up of 20 trinucleotides.
These sets are linked by the following  permutation property: $X_1 = \{\mathcal{P}(t), t \in X_0 \}$, $X_2 = \{\mathcal{P}(t), t \in X_1 \}$, where $\mathcal P$ is defined for all trinucleotide $t=n_0n_1n_2$ by $\mathcal{P}(t)=n_1n_2n_0$.
Additionally, if we denote $c:\mathcal{N} \longrightarrow \mathcal{N}$ the \emph{complementary function} defined on the set of nucleotides $\mathcal{N} = \{A,T,C,G\}$ by: $c(A)=T$, $c(T)=A$, $c(C)=G$,
$c(G)=C$, and for all words of nucleotides $u$ and $v$, $c(uv)=c(v)c(u)$, then we have $c(X_0)=X_0$,
$c(X_1)=X_2$, and $c(X_2)=X_2$, which is referred to the ``complementarity property''.
More details about the research context, the constitution of these sets, and their properties ($C^3$ code, rarity, largest window length, higher frequency of ``misplaced'' trinucleotides, flexibility) can be found in~\cite{Bahi2008,Bahi2008a}.
Among other things, it has been proven that $X_0$ occurs with the highest probability (48.8\%) in genes (reading frames 0), whereas $X_1$ and $X_2$ occur mainly in the frames 1 and 2 respectively.
In other words, $X_0$ is not pure in the reading frame (its probability is less than 1): it is mixed with $X_1$ and $X_2$.
Such a property has been explained by authors of~\cite{Bahi2008,Bahi2008a} as follows.
Suppose that $X_0$ represents the set of trinucleotides used to build the gene of the last common ancestor of the considered set of species.
Random mutations have introduced noise during evolution, leading to a decreased probability of $X_0$ \cite{Bahi2008,Bahi2008a}.
Another fundamental property is asymmetry in the sense that codes $X_1$ and $X_2$ satisfy $P(X_1) < P(X_2)$.
The standard and chaotic models (with particular strategies for the stop codons) can explain both the decreased probability of the code $X_0$ and the asymmetry between the codes $X_1$ and $X_2$ in genes, by
the following procedure.
Construct the ``primitive'' genes, i.e., genes before random substitutions, with trinucleotides of the circular code $X_0$.
Starting from this initial condition, the systems (standard or chaotic) are launched, iterating their
processes until a stop condition is checked.
By doing so, and for rates chosen carefully, it is possible to be close to the current frequency of each of the three codes $X_0$, $X_1$, and $X_2$ in genes.
In this situation, CM models largely outperform the standard models, being closer to the observed
probabilities for $X_0$, $X_1$, and $X_2$ discussed above.
In particular, the chaotic model called ``$CM_{TAA}$'' with low mutability of the stop codon TAA, matches as much as possible the probability discrepancy between the circular codes $X_1$ and $X_2$ observed in reality.
For further details, the reader is referred to~\cite{Bahi2008,Bahi2008a}.

All the properties described before show that the gene mutation prediction is  %strongly linked to chaotic properties, so chaotic models are 
suitable to describe these phenomena. This kind of manifestation of chaos in genomics is somewhat surprising and needs,
in our opinion, to be further investigated, determining whether more
fundamental reasons can justify the success of chaotic models to well simulate genome evolution.
In the following section, we will propose some reasons explaining why chaos is related to genomes.
More precisely, we will show that some genome evolution mechanisms, as modeled in the present article, are chaotic according to Devaney. To achieve this goal, we first need to recall the bases of the mathematical theory of chaos.

\section{Basic Remainders}

Let us now rigorously introduce the notions of Devaney's chaos and of chaotic iterations, with their respective links.
\subsection{Devaney's chaotic dynamical systems}
\label{sec:Devaney}
Consider  a topological  space $(\mathcal{X},\tau)$  and  a continuous function $f : \mathcal{X} \to \mathcal{X}$.
\begin{definition}
\label{def:topologicalTransitivity}
Function $f$ is said to be \emph{topologically transitive} if,  for any pair of non empty open sets $U,V \subset \mathcal{X}$, there exists $k>0$ such that $f^k(U) \cap V \neq \emptyset$.
\end{definition}
\begin{definition}
The point $x \in \mathcal X$  is a \emph{periodic point} for
$f$  of period  $n\in \mathbb{N}^*$  if $f^{n}(x)=x$.
% The set of periodic points of $f$ is denoted $Per(f).$
\end{definition}
\begin{definition}
Function $f$ is said to be \emph{regular} on $(\mathcal{X}, \tau)$ if the set of periodic points for $f$  is dense in $\mathcal{X}$: for any point $x$ in $\mathcal{X}$, any neighborhood  of $x$ contains at least one periodic point.
\end{definition}
\begin{definition}
Function $f$ is said  to be \emph{chaotic} on $(\mathcal{X},\tau)$  if $f$ is
regular and topologically transitive.
\end{definition}
In cases where the topology $\tau$ can be described by a metric $d$, the   chaos   property  is   strongly   linked   to   the  notion   of ``sensitivity'', defined on a metric space $(\mathcal{X},d)$ by:
\begin{definition}
\label{sensitivity} Function $f$ has \emph{sensitive dependence on initial conditions}
if there  exists $\delta >0$  such that, for any  $x\in \mathcal{X}$
and  any  neighborhood  $V$  of  $x$,  there  exists  $y\in  V$  and
$n\geqslant 0$  such that $d\left(f^{n}(x),  f^{n}(y)\right) >\delta
$.
\noindent Then $\delta$ is called the \emph{constant of sensitivity} of $f$.
\end{definition}
Indeed, Banks  \emph{et al.}  have proven in~\cite{Banks92}  that when
$f$ is chaotic on $(\mathcal{X}, d)$, then $f$ has
the  property  of sensitive  dependence  on  initial conditions  (this
property was formerly  an element of the definition  of chaos). To sum
up, quoting Devaney in~\cite{devaney}, a chaotic dynamical system ``is
unpredictable   because  of  the   sensitive  dependence   on  initial
conditions. It cannot be broken down or simplified into two subsystems
which do not interact because  of topological transitivity. And in the
midst  of this  random behavior,  we nevertheless  have an  element of
regularity''.  Fundamentally   different  behaviors  are  consequently
possible and occur in an unpredictable way.

\subsection{Chaotic Iterations}
\label{sec:chaotic iterations}
%Let us consider  a \emph{system} with a finite  number $\mathsf{N} \in
%\mathds{N}^*$ of elements  (or \emph{cells}), so that each  cell has a
%Boolean  \emph{state}. A  sequence of  length $\mathsf{N}$  of Boolean
%states of  the cells  corresponds to a  particular \emph{state  of the
%system}. A sequence which  elements belong to $\llbracket 1;\mathsf{N}
%\rrbracket $ is called a \emph{strategy}. The set of all strategies is
%denoted by $\mathbb{S}.$
%Let $\mathcal{X}$ be a set, $\mathsf{N} \in \mathds{N}^*$, and
\begin{definition}
\label{Def:chaotic iterations}
Let $\mathcal{X}$ be a set, $\mathsf{N} \in \mathbb{N}^*$,
$f:\mathcal{X}^{\mathsf{N}}\longrightarrow  \mathcal{X}^{\mathsf{N}}$ be
a  function,  and  $S$  be  a  sequence of subsets of $\ldbrack 1;\mathsf{N}\rdbrack $ called a ``chaotic strategy''.  The  
\emph{chaotic iterations}     are   the sequence $(x^n)_{n\in\mathbb N}$ of elements of $\mathcal X^\mathsf{N}$  defined      by     $x^0\in
\mathcal{X}^{\mathsf{N}}$ and
\begin{equation*}    
\forall    n\in     \mathbb{N}^{\ast     },    \forall     i\in
\ldbrack1;\mathsf{N}\rdbrack ,x_i^n=\begin{cases}
x_i^{n-1} &  \text{ if  } i \notin S^n \\  
\left(f(x^{n-1})\right)_{i} & \text{ if } i \in S^n.
\end{cases}
\end{equation*}
\end{definition}
In other words, at the $n^{th}$ iteration, only the components of $S^{n}$ are
\textquotedblleft  iterated\textquotedblright .  Note  that the term  ``chaotic'' in  the name of  these iterations, has \emph{a priori} no link with the mathematical theory of chaos, which will be recalled in the next section.
However, it has been proven in~\cite{bg10:ij} that, \emph{for a large variety of
functions, chaotic iterations are indeed really chaotic}.

\section{Genomics Mutations as a Discrete Dynamical System}\label{sec:genmut}
\subsection{Presentation of the Problem}
We now ask whether the evolution
of a DNA sequence under evolution can be predicted or not.
In this section, we will more specifically focus on the following
questions. Firstly, given a genome (or any DNA
sequence) $G$ of
interest, and a more or less precise idea of
mutations that it will probably face in future
(for instance, some areas in the genomes are known to mutate
more frequently than other ones), is it possible
to infer a set of the more probable genomes
that can result, in the future, from this original
sequence $G$ after mutations? Second, given a
sequence known at the current generation
(say, at time $t^n$), is it
possible to determine what was the most probable
aspect of this sequence in the past (at time
$t^m, m<n$)? Thirdly, given two DNA sequences,
the second one being the result of some mutations
on the first one, is it possible to find the
mutation sequence that has changed the first
sequence in the second one (taking into account
the fact that a given nucleotide can mutate
several times).
Obviously, with no information about the mutation
rate and history of the considered DNA sequence,
this prediction is quite impossible. But what
happens if we can follow the DNA sequence over several
generations, learning by doing so information about
the possible form of its mutations sequence? For
instance, following a lineage of \emph{Escherichia
coli} during 40000 generations gives us a
lot of informations concerning the behavior of
mutations in the genomes of the considered
lineage. Is it possible to use this knowledge to
predict the genome of this lineage at generation number 45000 ?
In other words,
knowing the initial  DNA sequence $G^0$ at time
$t^0$ and the 40000 first terms of the mutations
sequence, can we predict the DNA sequence at time
$t^{45000}$?
With the knowledge of $G^0$ and the whole mutations
sequence  $S=(S^0, \hdots, S^{45000})$, the genome
$G^{45000}$ can be obtained without prediction, but
what happens to our ability to make a prediction
when using only the head
$(S^0, \hdots, S^{40000})$ of this sequence?
This head can be seen as an approximation of the
true mutations sequence $S$, and if the evolution
dynamics of the mutations is quite stable through
approximations, in the sense where a small perturbation 
at the origin yields a small perturbation at the end 
of the process, then this prediction makes sense.
To measure the stability of the mutation dynamics
through small errors or approximations, and the
capability to predict the evolution of genomes
under mutations, we must firstly write this
mutation operation as a dynamical system, provide
an accurate distance that corresponds to the
``approximation'' quoted below, and measure the
effects of our ignorance on the complete mutations
sequence on the prediction of genomes evolution.

\subsection{Formalization of DNA Mutation Evolution}
\label{formalization}
A genome having $\mathsf{N}$ nucleotides is formalized here as a sequence of
$\mathsf{N}$ integers belonging in $\left\{1,2,3,4\right\}$, where $1$
(resp. $2, 3,$ and $4$) refers to the adenine (resp. cytosine, guanine,
and thymine). The benefit of using integers $\left\{1,2,3,4\right\}$ instead of 
$\left\{A,C,G,T\right\}$ is justified by the construction of a metric for the 
mutation process (see Section \ref{subsec:metric}).
An evolution under nucleotide mutations of this genome is a sequence of
couples of $\ldbrack 1, \mathsf{N} \rdbrack \times \ldbrack 1,4
\rdbrack$, where we infer that:
\begin{itemize}
\item time has been divided into a sequence
$t^0, t^1, \hdots, t^n, \hdots$ such that at most
one mutation can occur between two time intervals,
\item the $i-$th couple of the mutation sequence is
equal to $(m,n)$ if and only if the $m-$th
nucleotide of the genome is replaced into the
nucleotide $n$. If the $m-$th nucleotide was $n$,
then no mutation has occurred at time $t^i$.
\end{itemize}
Such a sequence will be called ``mutations
sequence'' in the remainder of this document.
$\mathcal{S}_\mathsf{N} = \displaystyle{
\bigcup_{n \in \mathbb{N}} \left(\ldbrack 1,
\mathsf{N} \rdbrack \times \ldbrack 1, 4 \rdbrack \right)^n}$
will denote the (infinite) set of all possible
mutations (finite) sequences.
We introduce the phase space $\mathcal{X}_\mathsf{N} = \ldbrack 1, 4 \rdbrack^\mathsf{N}
\times \mathcal{S}_\mathsf{N}$ as the set of mutating genomes. It is constituted by couples
of points that store the information of a genome \emph{and} its future
evolution: the first coordinate of the couple is the current DNA sequence
whereas the second coordinate is the sequence of
mutations that will appear in the future (the problem
is that this sequence can only be, in
the best case, approximate concretely).
\begin{example}
\label{ex1}
For instance, the point $\left((1,1,2,1,3), \left((2,2),(2,3),(1,4) \right) \right)\in \mathcal{X}_5$
corresponds to the evolution $\left\{AACAG, ACCAG, AGCAG, TGCAG\right\}$:
the left coordinate $(1,1,2,1,3)$ means that we start with the sequence
$AACAG$, whereas the second coordinate $ \left((2,2),(2,3),(1,4) \right)$
explains that:
\begin{enumerate}
\item the first mutation $(2,2)$ is a substitution
of the second nucleotide by $C$,
\item the second mutation $(2,3)$ is a substitution of the second
nucleotide by $G$,
\item the third and last mutation $(1,4)$ refers to the substitution of the first nucleotide by $T$, which is designed here by $4$.
\end{enumerate}
\end{example}
Let us now introduce the \emph{initial} and \emph{shift} operators $i$ and $\sigma$ defined respectively by
\begin{equation*}
\begin{array}{cccc}
i: & \mathcal{S}_\mathsf{N} & \longrightarrow & \ldbrack 1, \mathsf{N} \rdbrack \times \ldbrack 1, 4 \rdbrack\\
 & (s^0, s^1, \hdots ) & \longmapsto & s^0
\end{array}
\end{equation*}
and
\begin{equation*}
\begin{array}{cccc}
\sigma : & \mathcal{S}_\mathsf{N} & \longrightarrow & \mathcal{S}_\mathsf{N}\\
 & (s^0, s^1, \hdots ) & \longmapsto & (s^1, s^2, \hdots).
\end{array}
\end{equation*}
The shift operator corresponds to the so-called symbolic dynamics, a well-studied mathematical example of chaotic dynamics~\cite{Formenti1998}.
With this material, the mutation operation $\mathcal M$ can be written as follows:
\begin{equation}
\label{defM}
\begin{array}{cccc}
\mathcal{M} : & \mathcal{X}_\mathsf{N} & \longrightarrow & \mathcal{X}_\mathsf{N}\\
 & \left(G_1, \hdots, G_\mathsf{N}), S\right) & \longmapsto & \left((G_1, \hdots,  G_{i(S)_1-1}, i(S)_2, G_{i(S)_1+1}, \hdots, G_\mathsf{N}),  \sigma(S)\right).
\end{array}
\end{equation}
In other words, the nucleotide at position $i(S)_1$ in the genome  $\left(G_1, \hdots, G_\mathsf{N}\right)$
is replaced by the nucleotide $i(S)_2$, and the first substitution $i(S)$ in the mutation
sequence $S$ is removed (as the mutation has already been achieved).
Thus the DNA evolution as the generations pass can finally be written
as the following discrete dynamical system:
\begin{equation}
\label{SDDmutation}
\left\{
\begin{array}{l}
X^0=(G^0, S) \in \mathcal{X}_\mathsf{N}\\
X^{n+1} = \mathcal{M}(X^n).
\end{array}
\right.
\end{equation}
\begin{example}
Let us consider Example~\ref{ex1} another time.
As stated before,
$X^0 = \left((1,1,2,1,3), \left((2,2),(2,3),(1,4)\right)\right) \in \mathcal{X}_5$.
Then $X^1= \mathcal{M}(X^0) =  ((1,2,2,1,3),$ $\left((2,3),(1,4)\right)) $,
$X^2= \mathcal{M}(X^1)
=  ((1,3,2,1,3),$ $\left((1,4)\right)) $,
and
$X^3= \mathcal{M}(X^2) =  ((4,3,2,1,3), \varnothing)$.
The last DNA sequence, obtained after 3 mutations (3 iterations of the dynamical system), is thus equal to $G^3 = X^3_1 = (4,3,2,1,3)$, which is $TGCAG$.
\end{example}
\subsection{A Metric for Mutation based Genomes Evolution}\label{subsec:metric}
A relevant metric must now be introduced in order to measure the correctness of the prediction, and to give consistency to the notion of approximation that has occurred several times in the previous section.
This distance must be defined on the set $\mathcal{X}_\mathsf{N}$, to measure how close is a predicted DNA evolution to the real one. It will be constructed as follows: given $X=(X_1,X_2),Y=(Y_1,Y_2) \in \mathcal{X}_\mathsf{N}$, the number $d(X,Y)$:
\begin{itemize}
\item has an integer part that computes the differences between the two DNA sequences $X_1$ (for instance, the predicted or approximated genome) and $Y_1$ (the real genome), that is, the number of nucleotides that do not correspond in the two genomes.
\item has a fractional part that must be as small as the evolution processes $X_2, Y_2$ will coincide for a long enough duration. More precisely, the $k-$th digit of $d(X,Y)$ will be equal to $0$ if and only if, after $k$ generations, the same position (nucleotide) will be changed in both $X_1$ and $Y_1$ genomes, and the same nucleotide is
inserted in each case.
\end{itemize}
Such requirements lead to the introduction of the
following function:
\begin{equation*}
\forall X,Y \in \mathcal{X}_\mathsf{N}, d(X,Y)
= d_G(X_1,Y_1) + d_S(X_2,Y_2)
\end{equation*}
where
\begin{equation*}
\begin{cases}
d_G(X_1,Y_1) = \displaystyle{\sum_{k=1}^{\mathsf{N}} \delta (X_1^k, Y_1^k)},\\
d_S(X_2, Y_2) = \displaystyle{\dfrac{9}{\mathsf{N}} \sum_{k=0}^\infty \dfrac{\mathcal F(X_2^k-Y_2^k)}{10^{k+1}}},
\end{cases}
\end{equation*}
where $\delta$ is the discrete metric on $\mathbb R$, that is, for $x,y\in\mathbb R$, $\delta(x,y) = \begin{cases}
1 \text{ if } x\neq y, \\
0 \text{ else,}
\end{cases}$ and $\mathcal F:\mathbb R^2 \longrightarrow \mathbb{R}^+$ is given by $\mathcal F(x_1,x_2) = |x_1| + \delta(0,x_2)$.
\begin{proposition}
Function $d$ is a metric on $\mathcal{X}_\mathsf{N}$.
\end{proposition}

\begin{proof}
The function $d_G$ is clearly a metric on $\ldbrack 1,4\rdbrack^\mathsf{N}$ as being the 1-product metric of the $\mathsf{N}$ metric spaces $(\ldbrack 1,4\rdbrack,\delta)$.
We now prove that $d_S$ is a metric. Firstly, $d_S$ is well defined since for $(X_2,Y_2) \in \mathcal S_\mathsf{N}$, one gets $\mathcal F(X_2^k-Y_2^k) \leq \mathsf{N}+2$ for every $k\in\mathbb N$, implying the convergence of the series in the definition of $d_S$. Coincidence axiom and symmetry being obvious, we only prove the subadditivity of $f_S$. If $x_2,y_2\in\mathbb R$ are such that $\delta(0,x_2-y_2) = 1$ then $x_2\neq y_2$ and for every $z_2\in\mathbb R$, either $x_2\neq z_2$ or $y_2\neq z_2$ and so $\delta(0,x_2-z_2) + \delta(0,z_2-y_2) \geq 1$. Consequently for every $x,y,z\in\mathbb R^2$, $\mathcal F(x - y) \leq \mathcal F(x-z) + \mathcal F(z-y)$. The series being convergent for every $X_2,Y_2\in\mathcal S_\mathsf{N}$, one deduces that $d_G$ satisfies the subadditivity property on $\mathcal S_\mathsf{N}$ and is a metric on this set. As a consequence, $d$ is a metric on $\mathcal X_\mathsf{N}$.
\end{proof}

\subsection{The Topological Study of Mutations}

\subsubsection{Continuity}

Let us start by proving that,
\begin{proposition}
The mutation operation $\mathcal{M}$ is a continuous function on $(\mathcal{X}_\mathsf{N}, d)$.
\end{proposition}
\begin{proof}
This result will be established by using the sequential characterization of the continuity.
Let $(X^n)_{n \in \mathbb{N}}$ a sequence of $\mathcal{X}_\mathsf{N}$ that converges to $\ell \in \mathcal{X}_\mathsf{N}$. We must prove that $\mathcal{M}(X^n) \longrightarrow \mathcal{M}(\ell)$ in $(\mathcal{X}_\mathsf{N}, d)$.
Let $\varepsilon >0$. $X^n \longrightarrow \ell$ in $(\mathcal{X}_\mathsf{N}, d)$, then $\exists n_1 \in \mathbb{N}$, $\forall n \geqslant n_1$, $d(X^n, \ell ) \leqslant 0.09$.
So all $X^n$ for $n \geqslant n_1$ have the same first coordinate (genome), which is $\ell_1$. Furthermore, consequently to the definition of $d_S$, the first term of each $X_2^n$ for $n \geqslant n_1$ is equal to the first term of $\ell_2$. So, $\forall n \geqslant n_1$, $\mathcal{M}(X^n)_1 = \mathcal{M}(\ell)_1$.\\
Let $k_0 = \lceil - log_{10}(\varepsilon) \rceil$. As $d(X^n, \ell) \longrightarrow 0$, $\exists n_2 \geqslant n_1$ such that, for $n \geqslant n_2$, $d_S(X^n_2, \ell_2) \leqslant \dfrac{1}{10^{k_0+1}}$, meaning that the sequences $X^n_2, n\geqslant n_2$ and $\ell_2$ start all with the same $k_0+1$ terms. As the operation of $\mathcal{M}$ on the second coordinate of points of $\mathcal{X}_\mathsf{N}$ is a shift of one term to the left, we conclude that $\forall n \geqslant n_2$, $d_G(\mathcal{M}(X^n)_1, \mathcal{M}(\ell)_1) = 0$ and $d_S(\mathcal{M}(X^n)_2, \mathcal{M}(\ell)_2) < \dfrac{1}{10^{k_0}} \leqslant \varepsilon$, and thus $\mathcal{M}(X^n) \longrightarrow \mathcal{M}(\ell)$ in $(\mathcal{X}_\mathsf{N}, d)$, which ends the proof.
\end{proof}
As $\mathcal{M}$ is continuous, we can thus study the chaotic behavior of the discrete dynamical system of Eq.~\ref{SDDmutation}.
\subsubsection{Chaotic Behavior of DNA Mutations}

We first prove that,
\begin{proposition}
$\mathcal{M}$ is topologically transitive on $(\mathcal{X}_\mathsf{N}, d)$.
\end{proposition}
\begin{proof}
Let $X=(G,S)$ and $\check{X} = (\check{G},\check{S})$ two points of $\mathcal{X}_\mathsf{N}$, and $\varepsilon >0$. We will find $n\in \mathbb{N}$ and a point $X'=(G',S')\in\mathcal{B}\left(X, \varepsilon \right)$, the open ball centered on $X$ with radius $\varepsilon$, such that $\mathcal{M}^n (X') = \check{X}$.
Let $k_0 = \lceil - log_{10}(\varepsilon) \rceil$. Thus any point of the form $(G, (S^0, S^1, \hdots, S^{k_0},$ $s^{k_0+1}, s^{k_0+2}, \hdots ))$, with $s^{k_0+1}, s^{k_0+2}, \hdots  \in \ldbrack 1, \mathsf{N} \rdbrack$, is in $\mathcal{B}\left(X, \varepsilon \right)$. Suppose that $G$ and $\check{G}$ have $m$ different nucleotides, in position $i_1, \hdots, i_m \in \ldbrack 1, \mathsf{N} \rdbrack$. Thus the point 
\begin{equation*}
X' = (G, (S^0, S^1, \hdots, S^{k_0}, (i_1, \check{G}_{i_1}), \hdots, (i_m, \check{G}_{i_m}), \check{S}^0, \check{S}^1, \hdots) \in \mathcal{B}\left(X, \varepsilon \right)
\end{equation*} 
is such that $\mathcal{M}^{k_0+m+1}(X')=\check{X}$, leading to the transitivity of $\mathcal{M}$.
\end{proof}
\begin{remark}
A stronger result than the topological transitivity has indeed been stated in the proof above. It is called \emph{strong transitivity} and is defined by: for all $X,Y \in \mathcal{X}$ and for all neighborhood $V$ of $X$, it exists $n \in \mathbb{N}$ and $X’\in V$ such that $\mathcal{M}^n(X’)=Y$. Obviously, the strong transitivity implies the transitivity property.
\end{remark}
We now prove that,
\begin{proposition}
$\mathcal{M}$ is regular on $(\mathcal{X}_\mathsf{N}, d)$.
\end{proposition}
\begin{proof}
Let $X \in \mathcal{X}_\mathsf{N}$ and $\varepsilon > 0$. We have to exhibit a periodic point $X'\in\mathcal{B}(X, \varepsilon)$.
Let $k_0 = \lceil - log_{10}(\varepsilon) \rceil$. Suppose that $X=(G,(S^0, S^1, \hdots)$, and that the genome $\mathcal{M}^{k_0}(X)_1$ differ from $m$ nucleotides of $G$ in position $i_1, \hdots, i_m \in \ldbrack 1, \mathsf{N}\rdbrack$. Then the point:
\begin{equation*}
X' = (G, (S^0, \hdots, S^{k_0}, (i_1, G_{i_1}), \hdots (i_m, G_{i_m}),S^0, \hdots, S^{k_0}, (i_1, G_{i_1}), \hdots (i_m, G_{i_m}),\hdots)
\end{equation*}
is a periodic point in  $\mathcal{B}(X, \varepsilon)$.
\end{proof}

Let us finally demonstrate that:
\begin{proposition}
The mutation operator $\mathcal{M}$ has sensitive dependence on initial condition, and its constant of sensitivity is equal to $\mathsf{N}+\dfrac{\lfloor \dfrac{\mathsf{N}}{2}\rfloor +1}{\mathsf{N}}$.
\end{proposition}
\begin{proof}
Let $X=(G,S)\in \mathcal{X}_\mathsf{N}$ and $\varepsilon >0$. Let $k_0 = \lceil - log_{10}(\varepsilon) \rceil$. Consider a finite sequence of nucleotides $(n_1,\dots,n_\mathsf{N})\in \ldbrack 1, 4 \rdbrack^\mathsf{N}$ such that for each $i \in  \ldbrack 1, \mathsf{N} \rdbrack$,  $n_i \neq \left(\mathcal{M}^{k_0+\mathsf{N}}(X)_1\right)_i$, and an infinite sequence $(s^j)_{j \in \mathbb{N}}$ such that for every $j\in\mathbb N$,
\begin{itemize}
\item $s^j_1 = \begin{cases}
\mathsf{N} \text{ if } \left(\mathcal{M}^{k_0+\mathsf{N}+j}(X)_2\right)_1 \leqslant \dfrac{\mathsf{N}}{2}, \\
1 \text{ else}.
\end{cases}$\\[5pt]
\item $s^j_2 = \begin{cases}
4 \text{ if } \left(\mathcal{M}^{k_0+\mathsf{N}+j}(X)_2\right)_2=1, \\
1 \text{ else}.
\end{cases}$
\end{itemize}
Then the point
\begin{equation*}
X' = (G, (S^0, \hdots, S^{k_0}, (1, n_1), \hdots (\mathsf{N}, n_\mathsf{N}),s^0, s^1, \hdots) \in\mathcal{B}(X,\varepsilon)
\end{equation*}
and is such that $d(\mathcal{M}^{k_0+\mathsf{N}}(X),\mathcal{M}^{k_0+\mathsf{N}}(X')) \geqslant \mathsf{N}+\dfrac{\lfloor \dfrac{\mathsf{N}}{2}\rfloor +1}{\mathsf{N}}$.
Due to the definition of $(n_1,\dots,n_N)$ and $(s^j)_{j\in\mathbb{N}}$, the infimum in the latter equality is optimal, and the distance cannot be enlarged systematically for the neighborhood of all points.
\end{proof}
The three previous propositions lead to the following result.
\begin{theorem}
Genome mutations as modeled by our approach have a chaotic behavior according to Devaney.
\end{theorem}

\subsection{Further Investigations}

\subsubsection{Quantitative properties}

Genomic mutations possess  the  instability property:
\begin{definition}
A dynamical  system $\left( \mathcal{X}, f\right)$ is  unstable if for
all  $x \in  \mathcal{X}$, the  orbit $\gamma_x:n  \in  \mathbb{N} \longmapsto
f^n(x)$ is  unstable in the following sense, 
\begin{equation*}
\exists \varepsilon  > 0,\, \forall
\delta>0,\, \exists y \in \mathcal{X}, \,  \exists n \in \mathbb{N}, \text{ s.t. }
d(x,y)<\delta   \text{ and } d\left(f^{n}(n),f^{n}(y)\right)   \geq
\varepsilon.
\end{equation*}
\end{definition}
This  property, which  is  implied by  sensitive  dependence on initial conditions, leads to the fact that in all neighborhoods of any
genome evolution $(G,S)$ there  are points that can be separated with distance bigger than $\varepsilon$ in the
future through mutations. 

Let us now recall another common quantitative measure of disorder.
\begin{definition}
A function $f$ is said to have the property of \emph{expansivity} if
\begin{equation*}
\exists \varepsilon >0,\forall x\neq y,\exists n\in \mathbb{N},d(f^{n}(x),f^{n}(y))\geqslant \varepsilon .
\end{equation*}
\end{definition}
Then $\varepsilon $ is the \emph{constant of expansivity} of $f$: an arbitrarily small error on any initial condition is \emph{always} amplified of $\varepsilon $.
Let us prove that,
\begin{theorem}
Mutation operator $\mathcal{M}$ is expansive and its constant of expansivity is at least equal to 1.
\end{theorem}
\begin{proof}
If $X_1 \neq Y_1$, then $d(X,Y)=d(\mathcal{M}^0(X),\mathcal{M}^0(Y))  \geqslant 1$.\\
Or else necessarily $X_2 \neq Y_2$. Let $n = \min \{k\in \mathbb{N}, \, X_2^k \neq Y_2^k\}$. Then $\forall k<n, \mathcal{M}^k(X)_1 = \mathcal{M}^k(Y)_1$ and $\mathcal{M}^n(X)_1 \neq \mathcal{M}^n(Y)_1$, so $d(\mathcal{M}^n(X),\mathcal{M}^n(Y))  \geqslant 1$.
\end{proof}

\subsubsection{Qualitative properties}

Firstly, the   topological transitivity property implies indecomposability~\cite{Ruette2001}.

\begin{definition} \label{def10}
A   dynamical   system   $\left(   \mathcal{X},  f\right)$   is   {\bf
indecomposable}  if it  is not  the  union of  two closed  sets $A,  B
\subset \mathcal{X}$ such that $f(A) \subset A, f(B) \subset B$.
\end{definition}

Hence, taking into account only a small part of a genome in the modeling process, in order to  simplify the 
complexity of the studied dynamics, takes away from us to a global vision of mutations. Moreover, we will prove that genomic mutations
are topologically mixing, which is a strong version of transitivity:

\begin{definition}
A discrete dynamical system is said to be \emph{topologically mixing} if and only if, for any couple of disjoint open sets $U, V \neq \varnothing$, $n_0 \in \mathbb{N}$ can be found so that $\forall n \geqslant n_0, f^n(U) \cap V \neq \varnothing$.
\end{definition}

We have the result,

\begin{theorem}
$(\mathcal{X}_{\mathsf{N}},\mathcal{M})$ is topologically mixing.
\end{theorem}

This property is an immediate consequence of the lemma below.

\begin{lemma}
For all open ball $\mathcal{B} \subset \mathcal{X}_{\mathsf{N}}$, there exists an integer $n\in \mathbb{N}^{*}$ such that 
$\mathcal{M}^{(n)}\left( \mathcal{B}\right) = \mathcal{X}_{\mathsf{N}}$, where $\mathcal{M}^{(n)}$ is the $n$-th composition of operator $\mathcal{M}$ defined in \eqref{defM}.
\end{lemma}
\begin{proof}
Let $\mathcal{B} = \mathcal{B}\left( \left( \left( G_{1}, \hdots, G_{\mathsf{N}}\right),(S^{0}, 
\hdots) \right), \varepsilon \right)$, $k_{0} = -log_{10}\left( |\varepsilon |\right)$,{}
and $\check{X} = \left(\left(\check{G}_{1}, \hdots, \check{G}_{\mathsf{N}}\right),(\check{S}^{0}, 
\hdots )\right) \in \mathcal{X}_{\mathsf{N}}$.

We define $X'= \left( \left( G_{1}, \hdots, G_{\mathsf{N}}\right),(S^{0}, \hdots, S^{k_{0}}, 
(1, \check{G}_{1}), \hdots, (\mathsf{N}, \check{G}_{\mathsf{N}}), \check{S}^{0}, \hdots \right)$.
This point is such that $X' \in \mathcal{B}$ and $\mathcal{M}^{(k_{0}+\mathsf{N})}(X') = \check{X}$.
\end{proof}

Mutations $\mathcal{M}$ satisfy the notion of chaos according to Knudsen too, which
is defined by~\cite{Knudsen1994a}:
\begin{definition}
A discrete dynamical system is chaotic according to Knudsen if:
\begin{itemize}
	\item it is sensitive to the initial conditions,
	\item there is a dense orbit.
\end{itemize}
\end{definition}

Let us prove that,

\begin{proposition}
    The mutation operator is chaotic according to Knudsen on $(\mathcal{X}_\mathsf{N}, d)$.
\end{proposition}

\begin{proof}
The sensitivity to the initial condition has yet been stated. Let us define a point $X$
on $\mathcal{X}_\mathsf{N}$ having a dense orbit under iterations of $\mathcal{M}$.
$\mathcal{X}_\mathsf{N} = \ldbrack 1, 4 \rdbrack^\mathsf{N} \times \mathcal{S}_\mathsf{N}$ 
being the Cartesian product of two countable sets, it is countably infinite too: there exists
a bijection $\sigma:\mathbb{N} \longrightarrow \mathcal{X}_\mathsf{N}$.
Let $\mathcal{G}:\mathcal{X}_\mathsf{N} \longrightarrow \ldbrack 1, 4 \rdbrack^\mathsf{N}$, 
$(G,S) \longmapsto G$ be the first projection.
Then $X$ can be defined as follows:
$$
\begin{array}{lll}
    X = &( (1,1, \hdots, 1), (&
    (1, \mathcal{G}(\sigma(0))_1),  (2, \mathcal{G}(\sigma(0))_2), \hdots, 
    (\mathsf{N}, \mathcal{G}(\sigma(0))_\mathsf{N}), \sigma(0)_2, \\
    &&    (1, \mathcal{G}(\sigma(1))_1),  (2, \mathcal{G}(\sigma(1))_2), \hdots, 
    (\mathsf{N}, \mathcal{G}(\sigma(1))_\mathsf{N}), \sigma(1)_2, \\
    &&\hdots))
\end{array}$$
This $X$ is such that $\forall Y \in \mathcal{X}_\mathsf{N}, \exists n_Y \in \mathbb{N}, \mathcal{M}^{n_Y}(Y) = X$,
which is stronger than the required density.
\end{proof}

To a certain extent, this notion of chaos is less restrictive than the one
of Devaney. More precisely, 
Devaney's chaos implies Knudsen's chaos in compact spaces~\cite{Formenti1998}.

\subsubsection{Discussion}

Conclusion of this study of mutations is that they present a chaotic behavior leading to the impossibility to qualify the long term effect of an error in predicting the location and frequency of mutations into genomes. In the worst case scenario, this error will be amplified until having a completely different genome (all the nucleotides are different, as the constant of sensitivity is greater than the length of the genome). However this case is rather marginal, mutations do not occur as frequently as the generations pass, and a mutation implies a change of only one nucleotide, leading to the opinion that, at least in the short term, the general aspect of the genome under consideration still remains under control when only mutations occur.

Inversion and transpositions are another genomics rearrangements that mostly affects more than one nucleotide. Thus an error in the prediction of these operations can potentially more largely impact  the genome evolution. %This is why their dynamics will be studied in the next section, 
To qualify such impact, %.
%\section{Investigating the Dynamics of two other Genomics Rearrangements}
%\label{sec:other}
%\subsection{The case of inversions}
we first give some definitions useful to formalize inversions.
\begin{definition}
The complementary function $c:\ldbrack 1, 4 \rdbrack \longrightarrow \ldbrack 1, 4 \rdbrack$ is defined by $c(1)=4$, $c(4)=1$, $c(2)=3$, and $c(3)=2$.
\end{definition}
Then the complement of adenine A is thymine T, and $c(2)=3$ means, for instance, that the complement of cytosine is guanine. We can now define the inversion process on a chromosome:
\begin{definition}
Let $\mathsf{N} \in \mathbb{N}^*$, and $(n_1, \hdots, n_\mathsf{N})$  a chromosome. \emph{Inversions} have the form:
\begin{flushleft}
$(n_1, \hdots, n_{i-1}, \underline{n_i, n_{i+1}\hdots, n_{j-1}, n_j}, n_{j+1}, \hdots, n_\mathsf{N}) \longrightarrow$
\end{flushleft}
\begin{flushright}
$(n_1, \hdots, n_{i-1}, \underline{c(n_j), c(n_{j-1}) \hdots, c(n_{i+1}), c(n_i)}, n_{j+1}, \hdots, n_\mathsf{N}).$
\end{flushright}
\end{definition}
\begin{example}
For instance, $ACCT\underline{GTAA}TGTTA$ is a possible inversion of \linebreak $ACCT\underline{TTAC}TGTTA$.
\end{example}

Obviously, it is impossible to map the DNA sequence $AAAAAAAA$ into $CCCCCCCC$ using only inversions, as the complement of $A$ is $T$. This fact is in contradiction with the property of transitivity, leading to the statement that,
\begin{proposition}
The inversion rearrangement is not chaotic on the set of all genomes of size $\mathsf{N}$.
\end{proposition}

Let us finally investigate the dynamics of transposition inside genomes. 
Transposons are DNA sequences that can move into a given genome following a cut and paste mechanism of transposition:
\begin{flushleft}
$(n_1, \hdots, n_{i-1}, \underline{n_i, \hdots, n_j}, n_{j+1}, \hdots, n_k, n_{k+1}, \hdots, n_\mathsf{N}) \longrightarrow$
\end{flushleft}

\begin{flushright}
$(n_1, \hdots, n_{i-1}, n_{j+1}, \hdots, n_k, \underline{n_i, \hdots n_j}, n_{k+1},\hdots, n_\mathsf{N}).$
\end{flushright}

Obviously this transposition cannot fit the requirements of transitivity, as the number of 
adenines, thymines, guanines, and cytosines are preserved. Then, for instance, it is impossible
to join a genome with an high rate of thymine, starting transpositions on a genome with a
low rate of $T$. Thus,

\begin{proposition}
Transposition of transposons is not chaotic according to Devaney.
\end{proposition}

\section{Conclusion}

In this document, the three operations of genomics rearrangement that
can be modeled by discrete dynamical systems (due to the preservation
of the size of the genomes) have been studied using mathematical 
topology. It has been stated that mutations are chaotic, whereas
transpositions and inversions are not. The proposed models lead to the 
feeling that genome evolution generates moderate chaos, and
that this evolution can probably be predicted to a certain extent.

This claim will be further investigated in our future work, by making
a larger and complete study of all the possible rearrangements into
genomes, measure and study their frequency using the related literature, 
and discussing to what extend this prediction can be realized. 
In particular, authors will study the set of mutations, transpositions,
and inversion strategies, to take into account for the presence of 
recombination hotspots.

\bibliographystyle{plain}
\bibliography{mabase}
\end{document}